\documentclass[conference,a4paper]{IEEEtran}
\usepackage{amsmath}
\usepackage{amssymb}
\usepackage{amsfonts}
\usepackage{epsfig}
\usepackage{calc}
\usepackage{color}
\usepackage[all]{xy}
\usepackage{bm}
\usepackage{caption}
\usepackage{enumerate}

\newtheorem{theorem}{Theorem}[section]

\newtheorem{defn}{Definition}
\newtheorem{thm}{Theorem}[section]
\newtheorem{cor}[thm]{Corollary}
\newtheorem{prop}{Proposition}

\newtheorem{lem}[thm]{Lemma}
\newtheorem{conj}[thm]{Conjecture}
\newtheorem{constr}[thm]{Construction}
\newtheorem{note}{Remark}
\newtheorem{example}{Example}
\newcommand{\bit}{\begin{itemize}}
	\newcommand{\eit}{\end{itemize}}
\newcommand{\bcor}{\begin{cor}}
	\newcommand{\ecor}{\end{cor}}
\newcommand{\beq}{\begin{equation}}
\newcommand{\eeq}{\end{equation}}
\newcommand{\beqn}{\begin{equation*}}
\newcommand{\eeqn}{\end{equation*}}
\newcommand{\bea}{\begin{eqnarray}}
\newcommand{\eea}{\end{eqnarray}}
\newcommand{\bean}{\begin{eqnarray*}}
	\newcommand{\eean}{\end{eqnarray*}}
\newcommand{\ben}{\begin{enumerate}}
	\newcommand{\een}{\end{enumerate}}
\newcommand{\bdefn}{\begin{defn}}
	\newcommand{\edefn}{\end{defn}}
\newcommand{\bnote}{\begin{note}}
	\newcommand{\enote}{\end{note}}
\newcommand{\bprop}{\begin{prop}}
	\newcommand{\eprop}{\end{prop}}
\newcommand{\blem}{\begin{lem}}
	\newcommand{\elem}{\end{lem}}
\newcommand{\bthm}{\begin{thm}}
	\newcommand{\ethm}{\end{thm}}
\newcommand{\bconj}{\begin{conj}}
	\newcommand{\econj}{\end{conj}}
\newcommand{\bconstr}{\begin{constr}}
	\newcommand{\econstr}{\end{constr}}
\newcommand{\bpf}{\begin{proof}}
	\newcommand{\epf}{\end{proof}}

\begin{document}
\title{An Alternate Construction of an Access-Optimal Regenerating Code with Optimal Sub-Packetization Level} 
\author{Gaurav Kumar Agarwal, Birenjith Sasidharan and P. Vijay Kumar \\
 	Department of ECE, Indian Institute of Science, Bangalore, 560012 India \\
 	(email: \{agarwal, biren, vijay\}@ece.iisc.ernet.in) 	
 }

\maketitle

\begin{abstract}
Given the scale of today's distributed storage systems, the failure of an individual node is a common phenomenon.  Various metrics have been proposed to measure the efficacy of the repair of a failed node, such as the amount of data download needed to repair (also known as the repair bandwidth), the amount of data accessed at the helper nodes, and the number of helper nodes contacted. Clearly, the amount of data accessed can never be smaller than the repair bandwidth. In the case of a help-by-transfer code, the amount of data accessed is equal to the repair bandwidth. It follows that a help-by-transfer code possessing optimal repair bandwidth is access optimal. The focus of the present paper is on help-by-transfer codes that employ minimum possible bandwidth to repair the systematic nodes and are thus access optimal for the repair of a systematic node. 

The zigzag construction by Tamo et al. in which both systematic and parity nodes are repaired is access optimal. But the sub-packetization level required is $r^k$ where $r$ is the number of parities and $k$ is the number of systematic nodes. To date, the best known achievable sub-packetization level for access-optimal codes is $r^{k/r}$ in a MISER-code-based construction by Cadambe et al. in which only the systematic nodes are repaired and where the location of symbols transmitted by a helper node depends only on the failed node and is the same for all helper nodes.  Under this set-up, it turns out that this sub-packetization level cannot be improved upon. In the present paper, we present an alternate construction under the same setup, of an access-optimal code repairing systematic nodes, that is inspired by the zigzag code construction and that also achieves a sub-packetization level of $r^{k/r}$. 
\end{abstract}

\begin{IEEEkeywords} Distributed storage, array codes, access-optimal, regenerating codes, sub-packetization.
\end{IEEEkeywords}

\section{Introduction}

In a distributed storage system, the data file comprising of $B$ data symbols drawn from a finite field $\mathbb{F}_q$, is encoded using  an error-correcting code of block length $n$ and the resulting code symbols are respectively stored in $n$ nodes of the storage network.  A naive strategy aimed at achieving resilience against node failures is to store multiple replicas of the same data.   In an effort to reduce the storage overhead, given the massive amount of data that is currently being stored, sophisticated codes such as Reed-Solomon codes are being employed in practice. Quite apart from resiliency to node failure with reduced storage overhead, there are several other attributes that are desirable in a distributed storage system.  These include: 
\bit
\item small repair bandwidth, i.e., the amount of data download in the case of a node failure is much smaller in comparison with the file size $B$, 
\item low repair degree, i.e., the number of helper nodes contacted for node repair is small, 
\eit

In \cite{dimakis_intro}, the regenerating-code framework was introduced, which addresses the problem of reducing the repair bandwidth. In an $(n,k,d)$-regenerating code, each of the $n$ nodes in the network stores $\alpha$ code symbols drawn from a finite field $\mathbb{F}_q$. The parameter $\alpha$ is termed as the sub-packetization level of the code. A data collector can download the data by connecting to any $k$ nodes and node repair  is accomplished by connecting to any $d$ nodes and downloading $\beta \leq \alpha$ symbols from each node with $\alpha \leq d \beta << B$. Thus $d\beta$ is the repair bandwidth.  

Here one makes a distinction between functional and exact repair.  By functional repair, it is meant that a failed node will be replaced by a new node such that the resulting network continues to satisfy the data collection and node-repair properties defining a regenerating code.   An alternative to function repair is {\em exact repair} under which one demands that the replacement node store precisely the same content as the failed node.  From a practical perspective, exact repair is clearly preferred. A cut-set bound based on network-coding concepts, tells us that under functional repair, given code parameters $(n,k,d, (\alpha,\beta))$, the maximum possible size of a data file is upper bounded~\cite{dimakis_intro} by 
\bea \label{eq:cut_set_bd}
B & \leq & \sum_{i=1}^{k} \min\{\alpha,(d-i+1)\beta\} .
\eea
Furthermore, this bound has been shown to be tight using network-coding arguments related to multicasting under functional repair. For fixed values of $(n,k,d,B)$, the bound in \eqref{eq:cut_set_bd} characterizes a tradeoff between $\alpha$ and $\beta$, referred to as the Storage-Repair Bandwidth tradeoff. The two extremal points in the tradeoff are respectively, the minimum-storage regenerating (MSR) and minimum bandwidth regenerating (MBR) points which correspond to the points at which the storage and repair bandwidth are respectively minimized. At MBR point, we have
\beq
\alpha \ = \ d\beta, \ B \ = \ k\alpha - {k \choose 2}\beta,
\eeq
and at MSR point, we have
\beq
\alpha \ = \ (d-k+1)\beta, \ B \ = \ k\alpha.
\eeq
It may be noted that MSR codes are Maximal-Distance Separable (MDS)\footnote{Unless otherwise mentioned, by an MDS code, we will mean a vector MDS code, i.e., an MDS code with a vector symbol alphabet.} in nature since $B=k\alpha$. Several exact-repair codes can be found in the literature that achieve the MBR and MSR points. There are a few constructions of MDS codes in literature that repair systematic nodes downloading the minimum bandwidth of $\frac{d\alpha}{d-k+1}$. In this paper, we focus on exact-repair MDS codes that achieve optimal bandwidth while repairing any systematic node. Tamo et al. \cite{optimal_access_tamo} proposed an MSR code for any $(n,k,d=n-1)$, referred to zigzag codes, that requires a sub-packetization level of $r^{k+1}$, where $r:=n-k$. 

Also of practical interest in a regenerating code, is the number of symbols accessed in each of the helper nodes, en route to computing the $\beta$ symbols to be transferred from the particular helper node to the failed node.   Clearly, this number cannot be less than $\beta$ and in instances where it is equal to $\beta$, the code is said to be {\em access-optimal}. zigzag codes have been shown to be access-optimal. In \cite{cadambe2011polynomial}, Cadambe et al. gave constructions of access-optimal MDS codes, that optimally repair the systematic nodes. The Cadambe et al. construction builds on the construction of the MISER code \cite{miser_code}, and requires a sub-packetization level of $r^{k/r}$. In \cite{access_bandwidth}, Tamo et al. showed that the sub-packetization level of an access-optimal MDS code, that optimally repairs the systematic nodes, is lower bounded by $r^{k/r}$ under the additional proviso that the location of symbols transmitted by a helper node depends only upon the failed node and is the same for all helper nodes.   Thus the construction in \cite{cadambe2011polynomial} is optimal in terms of access. 

In this paper, we will give an alternate construction of an access-optimal MDS regenerating code having parameters $(n,k,d=n-1)$ that optimally repairs every systematic node. Our construction is motivated by the zigzag code construction, but employs a novel sequential strategy for repair of a failed node.  

\section{Two Example Code Constructions} \label{code_construction}

We will first illustrate the construction using two examples. 

\begin{example} Let $(n,k,d) = (6,4,5)$, so that $r=(n-k)=2$.  
	Here we set $\alpha = r^{\frac{k}{r}} \ = \ 2^{\frac{4}{2}} = 4$. For the code to be an access-optimal MDS code with optimal bandwidth for repair of systematic nodes, we need to satisfy the following conditions:
	\begin{itemize}
		\item[(i)]	$B = k\alpha = 16$ 
		\item[(ii)] One should be able to reconstruct all the data accessing any 4 nodes. 
		\item[(iii)] It must be possible to repair a failed systematic node by accessing $\beta = \frac{\alpha}{r} = \frac{(4)}{2} = 2$ symbols from remaining $d= 5$ nodes.
	\end{itemize} 
	
	\begin{figure}[ht!]
		\centering
		\includegraphics[width=3.0in]{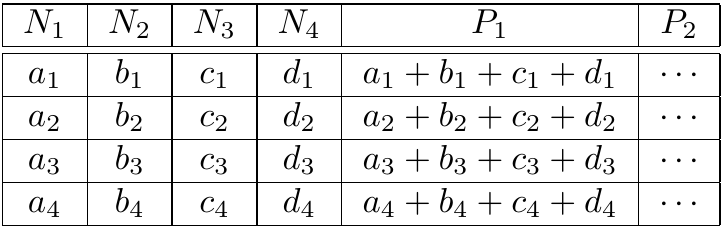}
		\caption{$P_1$ entries the codeword array for $k= 4, r=2, \alpha = 4$}
		\label{fig:k_4_r_2}
	\end{figure}
	
	In Fig.~\ref{fig:k_4_r_2}, $N_1$ to $N_4$ are systematic nodes while $P_1$ and $P_2$ are parity nodes. Clearly, the file size is 16 symbols. In the construction, the first parity (i.e., $P_1$) will always denote row parity. The remaining parities are designed to meet {\em Condition (ii)} and {\em Condition (iii)}.  In the present case of the example, there is only one remaining parity, i.e., $P_2$. Since for optimal repair, $\beta = \frac{\alpha}{2}$, we can access exactly $2$ symbols from each of the $5$ nodes.
	
	Let us define an index set $G = \{1, 2, \ldots, \alpha\}$.   The $i$th symbols in every node is indexed by the $i$th element of $G$. In our case, $G = \{1,2,3,4\}$. We split $G$ into two sets of equal size, so that each contains $2$ elements. After the splitting, we obtain two sets $G_1$, $G_2$. In the present instance, $G_1=\{1,2\}$ and $G_2=\{3,4\}$. Now we will further divide each of these sets into two sets of equal size: i.e., to split $G_1$ into $G_{11}=\{1\}$,$G_{12}=\{2\}$ and $G_2$ into $G_{21}=\{3\}$, $G_{22}=\{4\}$. We then form $G_3= G_{11} \cup G_{21} = \{1,3\}$ and $G_4= G_{12} \cup G_{22} = \{2,4\}$. Since the sets $G_{11}$, $G_{12}$, $G_{12}$, $G_{22}$ are singleton sets, no further splitting is possible and the procedure ends here. At this point, we have four sets in hand $G_1$, $G_2$, $G_3$ and $G_4$. The splitting procedure is shown in Fig.~\ref{example1_graph_division}. 
	
	\begin{figure}[ht!]
		\centering
		\includegraphics[width=3.0in]{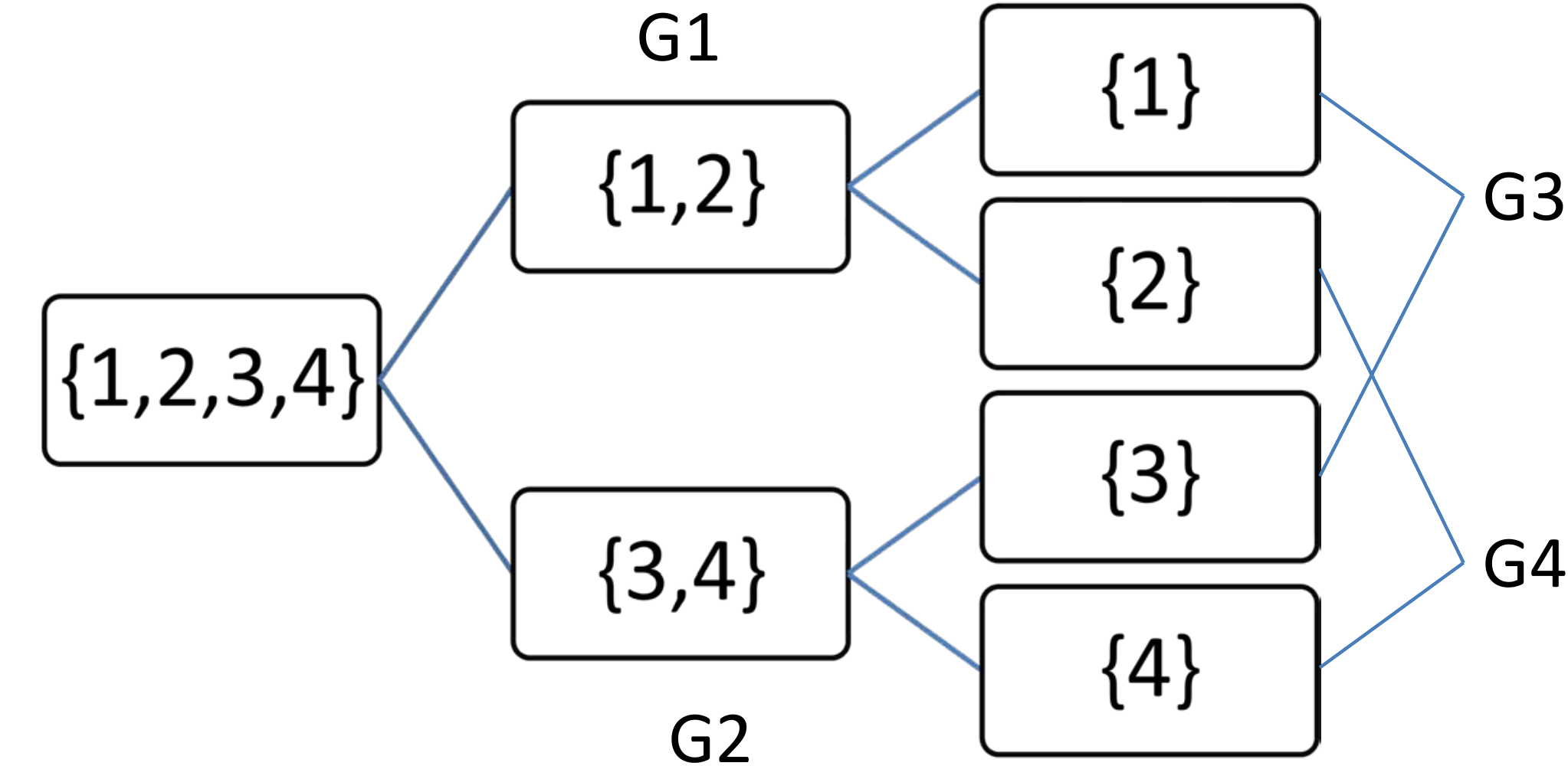}
		\caption{Splitting index set $G$ into various sets}
		\label{example1_graph_division}
	\end{figure}
	
	The sets $G_1, G_2, G_3, G_4$ correspond to the indices of the symbols of helper nodes, to be accessed while repairing the systematic nodes $N_1$, $N_2$, $N_3$ and $N_4$ respectively. For example, if $N_3$ fails, we will access symbols indexed by $G_3=\{1,3\}$ of remaining nodes. This completes the description of the repair strategy. Note that we are accessing only the optimal number $2$ of symbols from the helper nodes in accordance with {\em Condition (iii)}.   We will next show how the parity symbols belonging to node $P_2$ are computed.  
	
	Consider repair of the systematic node $N_1$. The repair will be carried out by accessing symbols $1$ and $2$ of the remaining nodes. It is clear from Fig.~\ref{fig:k_4_r_2} that, even if we do not access symbols from $P_2$, we can recover the first and second symbols of node $N_1$ i.e., $a_1$ and $a_2$. After repairing $a_1$ and $a_2$, we have access to \{$a_1$, $b_1$, $c_1$, $d_1$, $a_2$, $b_2$, $c_2$, $d_2$\}. Our goal is to obtain $a_3$ and $a_4$ using the first and second symbols of $P_2$. This requirement places a constraint on the first and second symbols of $P_2$. The first and second symbols of $P_2$ must form a set of two independent linear combinations of $a_3$, $a_4$, possibly along with linear combinations of symbols having subscripts $1$ or $2$. If we place these constraints on the symbols of $P_2$ taking into account, the repair of all $4$ systematic nodes, we will obtain the structure shown in Fig.~\ref{k_4_r_2_parity_step1}. In Fig.~\ref{k_4_r_2_parity_step1}, the set $\{i,j\}$ is short-hand notation for the collection of all message symbols having subscripts $i$ and $j$. 
	
	\begin{figure}[ht!]
		\centering
		\includegraphics[width=3.0in]{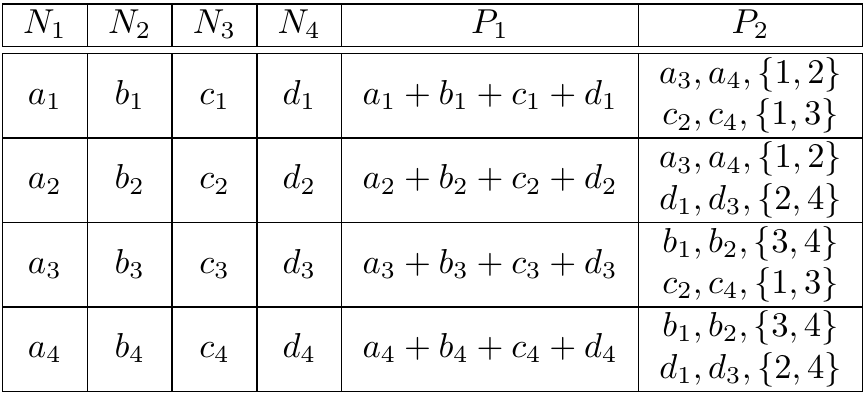}
		\caption{Designing symbols in $P_2$: STEP I}
		\label{k_4_r_2_parity_step1}
	\end{figure}
	
	In Fig.~\ref{k_4_r_2_parity_step1}, each cell in $P_2$ contains two lines: the first line corresponding to constraints arising out of repair scenario of $N_1$ or $N_2$; the second line corresponding to constraint of arising out of repair scenario of $N_3$ or $N_4$. Since both constraints must be satisfied, we have to take intersection of these two constraints. This leads to parity constraints as shown in Fig.~\ref{k_4_r_2_parity_step2}.
	
	\begin{figure}[ht!]
		\centering
		\includegraphics[width=3.0in]{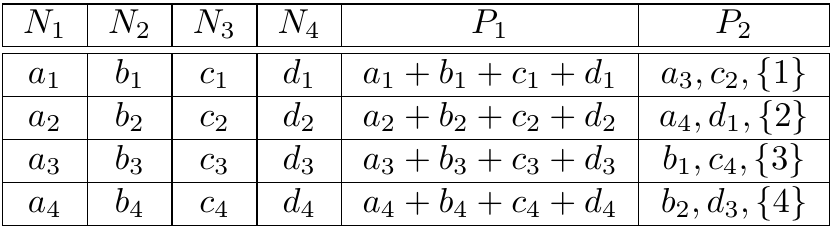}
		\caption{Designing symbols in $P_2$: STEP II }
		\label{k_4_r_2_parity_step2}
	\end{figure}
	
	So far we have identified the message symbols, whose linear combination leads to the parity symbols in $P_2$. This ensures the repair of systematic nodes, but will not guarantee the vector MDS property of the code. This will be ensured by choosing appropriate coefficients while making the linear combinations. Existence of such a choice of coefficients will be proved in Sec.~\ref{mds_property}. In this example, if we choose coefficients as shown in Fig.~\ref{k_4_r_2_parity_step3}, we can satisfy the MDS property.
	
	\begin{figure}[ht!]
		\centering
		\includegraphics[width=3.0in]{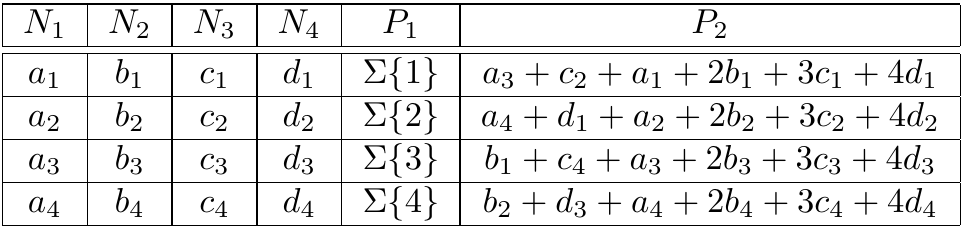}
		\caption{Codeword array for $k=4,r = 2,\alpha = 4$}
		\label{k_4_r_2_parity_step3}
	\end{figure}
	
\end{example}

\begin{example}	$(n,k,d) = (9,6,8)$
	Here $\alpha = 3^{\frac{6}{3}} = 9$. As in the previous example, we need to satisfy the following conditions:
	
	\begin{itemize}
		\item[(a)] $B = k\alpha = 54$ 
		\item[(b)] One should be able to reconstruct all the data by accessing any $6$ nodes. 
		\item[(c)] It must be possible to repair a failed systematic node by accessing $\beta = \frac{\alpha}{r} = \frac{9}{3} = 3$ symbols from the remaining $d= 8$ nodes. 
	\end{itemize}
	
	\begin{figure}[ht!]
		\centering
		\includegraphics[width=3.0in]{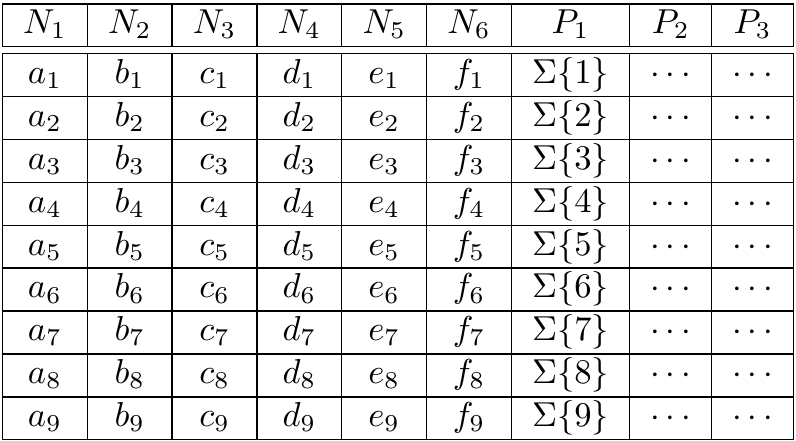}
		\caption{$P_1$ in codeword array for $k= 6,r = 3,\alpha=9$}
		\label{k_6_r_3}
	\end{figure}
	
	In this example as shown in Fig.~\ref{k_6_r_3}, $N_1$ to $N_6$ are the systematic nodes while $P_1$, $P_2$ and $P_3$ are parity nodes. Clearly, the file size is $54$. As in the previous example, $P_1$ represents row parity, while parity symbols $P_2$ and $P_3$ must be chosen in such a way that {\em Condition (b)} and {\em Condition (c)} are satisfied. Since for optimal repair, $\beta = \frac{\alpha}{3}$, we are permitted to access $3$ symbols from each of the $8$ nodes while repairing a systematic node.

	Since $\alpha=9$ in the present example, we set the index set $H = \{1, 2, \ldots, 9\}$ and symbols in each node are indexed by the elements of $H$. We will first split $H$ into three sets of equal size, each containing $3$ elements. After this splitting, we obtain the three sets $H_1=\{1,2,3\}$, $H_2=\{4,5,6\}$ and $H_3=\{7,8,9\}$. Next, we further divide each of these sets into three sets of equal size: i.e., we further split $H_1$ into $H_{11}=\{1\},H_{12}=\{2\}, H_{13}=\{3\}$, $H_2$ into $H_{21}=\{4\}, H_{22}=\{5\}, H_{23}=\{6\}$ and $H_3$ into $H_{31} = \{7\}, H_{32} = \{8\}, H_{33} = \{9\}$.   As in the previous example, in the third step, we form $H_4= H_{11} \cup H_{21} \cup H_{31} = \{1,4,7\}$, $H_5= H_{12} \cup H_{22} \cup H_{32} = \{2,5,8\}$ and $H_6= H_{13} \cup H_{23} \cup H_{33} = \{3,6,9\}$.  Since the sets $H_{ij}, i \in \{1,2,3\}, j\in \{1,2,3\}$ cannot be further divided, the procedure ends here. At the conclusion of this process, we have six sets in hand, namely $H_1$ through $H_6$.  
	
	Again, as in the case of the previous example, the sets $H_1$ through $H_6$ identify the indices of the symbols of the helper nodes to be accessed while repairing the systematic nodes $N_1$ to $N_6$ respectively. For example, if node $N_3$ fails, we will access symbols of the remaining nodes, indexed by the elements of $H_3=\{7,8,9\}$.   At this point, we have specified which symbols are transferred by a helper node in the case of failure of each of the $6$ systematic nodes.  We will next specify the contents of the parity nodes $P_2$ and $P_3$ and show that help-by-transfer as outlined above is indeed possible.   
		
	Consider repair of the systematic node $N_1$.  Since $H_1\ = \ \{1,2,3\}$, repair will be carried out by accessing symbols $1$, $2$ and $3$ of the remaining nodes. It is clear from Fig.~\ref{k_6_r_3} that the contents of the row-parity node $P_1$ and the remaining systemic nodes $N_2,N_3,\cdots,N_6$ accessed, suffice to repair the  first three symbols of $N_1$ i.e., $a_1$, $a_2$ and $a_3$. After repairing $\{a_1, a_2,a_3\}$, we have access to the message symbols with indices in $\{1,2,3\}$ from every systematic node including $N_1$. Our goal next is to recover $\{a_4, a_5 \ldots, a_9\}$ using the first, second and third symbols of $P_2$ and $P_3$. This requirement places a constraint on the first three symbols of $P_2$ and $P_3$. The first three symbols of $P_2$ and $P_3$ must be independent linear combinations of $\{a_4, a_5 \ldots, a_9\}$ along with linear combinations of symbols with indices lying in $\{1,2,3\}$. If we identify such constraints on the symbols of $P_2$ and $P_3$, while considering the repair of all $6$ systematic nodes, we will obtain the structure shown in Fig.~\ref{k_6_r_3_final}. In Fig.~\ref{k_6_r_3_final}, $\{i\}$ is shorthand notation for the collection of message symbols having index $i$.
	
	\begin{figure}[ht!]
		\centering
		\includegraphics[width=3.0in]{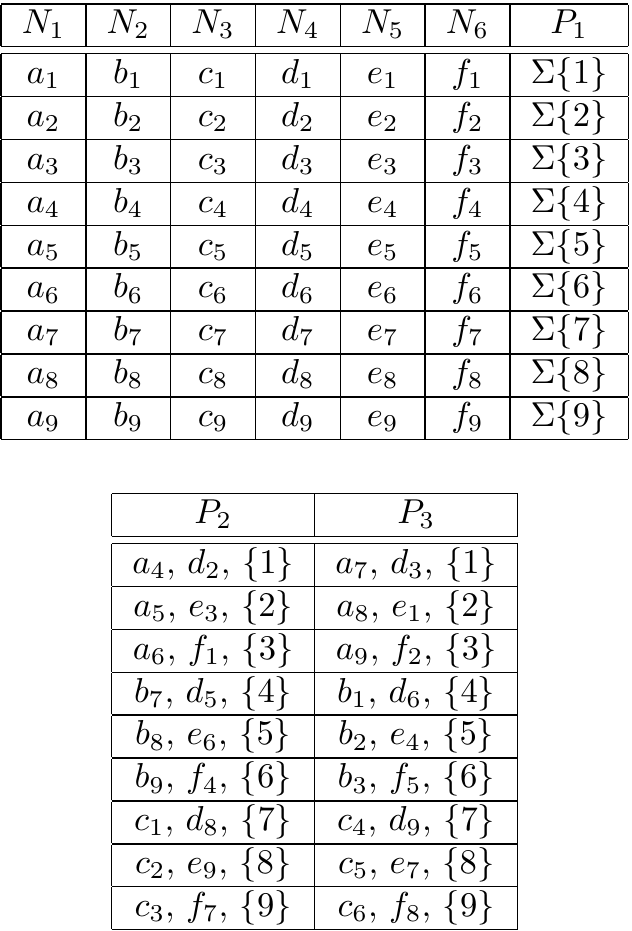}
		\caption{Codeword array $k= 6,r = 3,\alpha=9$}
		\label{k_6_r_3_final}
	\end{figure}
	
	We have identified thus far the message symbols whose linear combinations lead to the parity symbols in $P_2$ and $P_3$.   While this ensures repair of systematic nodes, it does not guarantee the MDS property of the code. This will be ensured by choosing appropriately, the coefficients which appear in the linear combinations. As will be shown in Sec.~\ref{mds_property}, a suitable set of coefficients can be found if one searches within a sufficiently large finite field.
	
	
\end{example}

\section{The General Construction for $(n,k,d=n-1)$} \label{general_code_construction}

The general construction assumes that the number $k$ of systematic nodes is a multiple of the number $r$ of parity nodes, i.e., $k \ = \ mr$ for some integer $m \geq 1$.   In the construction, the parameter $\alpha$ is given by  
\beqn
\alpha \ = \ r^m.
\eeqn
Hence the file size $B=mr^{m+1}$. We will represent each of the $mr$ systematic nodes by a $2$-tuple index $(s,t), s \in \{1,2,\ldots, m\}, t \in \mathbb{Z}_r:=\{0,1,\ldots, r-1\}$. Each of these nodes contains $\alpha = r^m$ symbols which we will index using the $m$-tuple $(y_1, y_2, \ldots ,y_m) \in \mathbb{Z}_r^m$. 

Suppose a systematic node $(s,t)$ fails. The repair strategy will then be such that each of the remaining nodes will then transmit symbols having index $(y_1, y_2, y_3, \ldots, y_m)$ with 
\beq \label{eq:symbol_choice}
y_s \ = \ t .
\eeq
Clearly, $r^{m-1}$ symbols from every node satisfy the constraint in \eqref{eq:symbol_choice}.  We note that the code is indeed a help-by-transfer code, and that further, the indices of the helper data transmitted are dependent only upon the failed node. Furthermore, $\beta = \frac{\alpha}{r}$ ensuring that the repair is access-optimal. We will next describe the encoding process used to determine the contents of the $r$ parity nodes and verify that the construction does indeed, meet the requirements of repair and result in an MDS code.  

Consider the $x$-th parity node, $x \in \{0,1,\ldots, r-1\}$. As in the case of a systematic node, each symbol of the parity node is also indexed by an $m$-tuple $\underline{f} = (f_1, f_2, \ldots ,f_m) \in \mathbb{Z}_r^m$. Our goal is to show how the parity symbol indexed by $\underline{f}$ is computed from the message data. Based on the repair strategy, $\underline{f}$ will help in repairing the $m$ systematic nodes indexed by $(i, f_i), i \in \{1, 2, \ldots, m\}$. Keeping this in mind, we first identify the message symbols from the systematic nodes whose linear combination yields the parity symbol having index $\underline{f}$. We define $R_1$ as the set of message symbols belonging to the systematic node $(i,f_i), i \in \{1,2,\ldots,m\}$ and that are indexed by $(f_1, f_2, \ldots, f_{i-1}, f_i+x, f_{i+1}, \ldots, f_m)$. Clearly $R_1$ has size $m$. We next define $R_2$ as the set of message symbols belonging to the systematic nodes that are indexed by $(f_1, f_2, \ldots ,f_m)$.  Clearly $R_2$ consists of $mr$ elements. Here we note that in the case of the $0$-th parity, $R_1 \subset R_2$, while for the rest of the parities $ R_1 \cap R_2 = \phi$. The parity symbol $\underline{f}$ is computed as a  linear combination of the symbols from $R_1 \cup R_2$, with coefficients lying in a sufficiently large field $\mathbb{F}_q$. Furthermore, each coefficient associated to a symbol from $R_1$ must be non-zero. The choice of non-zero coefficients to symbols from $R_1$ will ensure access-optimal repair of any systematic node. The additional freedom provided by the symbols from $R_2$ will turn out to be helpful in making the code MDS. This will be made clear in the next section. The feasibility of repair of any systematic node is stated in the following theorem:

\begin{theorem} Any failed systematic node $(s,t), s \in \{1,2,\ldots, m\}, t \in \mathbb{Z}_r$ can be repaired using the repair strategy mentioned above accessing $\beta= {r}^{m-1}$ symbols from each of the remaining $(m+1)r-1$ nodes.
\end{theorem}

\bpf First, consider the subset of symbols belonging to the systematic node $(s,t)$ having indices lying in the set $T_0$ given by 
\beq
T_0 = \{ (y_1,y_2, \ldots, y_m) \mid y_s = t \}.
\eeq 
As per the repair strategy outlined above, we will have access to $r^{m-1}$ symbols from each of the remaining systematic nodes, namely, those symbols whose symbol indices satisfy $y_s = t$. Let us denote this set of message symbols by $S$. In addition, we have access to the set $H_0$ of $r^{m-1}$ parity symbols from the $0$-th parity node whose indices satisfy $y_s=t$. With the aid of the elements in $S \cup H_0$, we will be able to repair the symbols in the failed systematic node $(s,t)$, indexed by the set $T_0$. Let us denote the set of these symbols by $M(T_0)$.

Next consider the subset of symbols from $(s,t)$
\beq
T_j = \{ (y_1,y_2, \ldots, y_m) \mid y_s = t+j \}
\eeq for a fixed value of $j \in \mathbb{Z}_r^{+}$. Let us denote this subset of symbols by $M(T_j)$.
Note that we have access to the set $H_j$ of $r^{m-1}$ parity symbols from the $j$-th parity node whose indices satisfy $y_s=t$.  These parity symbols are formed as linear combination of symbols from $M(T_j)$ and $S \cup M(T_0)$. Hence with the help of symbols from $S \cup M(T_0) \cup H_j$, it is possible to repair $M(T_j)$. Since $\bigcup_{j \in \mathbb{Z}_r} M(T_j)$ covers the entire set of symbols in the node $(s,t)$, we are done.
\epf

In the next section, we will show that there exists an appropriate choice of coefficients for symbols from $R_1$ and $R_2$ that ensures the vector MDS property of the code.

\section{Proof of the MDS Property} \label{mds_property}


In the previous section, we have identified for each of the $\alpha$ symbols within a parity node, a set of message symbols $R_1 \cup R_2$ whose linear combinations yield the parity symbol, in such a way that access-optimal repair is possible. Note that the sets $R_1$ and $R_2$ vary depending on the particular parity symbol of interest. In the present section, we will show that we can always find an appropriate set of coefficients used in forming linear combinations of the message symbols in $R_1 \cup R_2$ that make the code an MDS code.  Note that in the description of the repair process, it was assumed that the coefficients attached to the symbols from $R_1$ were non-zero. In the case of parity symbols belonging to the $0$-th (row) parity node, since $R_1 \subset R_2$, it is sufficient that coefficients of symbols in $R_2$ be non-zero. For symbols from the $j$-th parity node $j \in \mathbb{Z}_r^+$, we have $R_1 \cap R_2 = \phi$, and hence it is sufficient that the coefficients of every symbol in $R_1$ have a fixed non-zero coefficient, say $c \neq 0$. The value of $c$ is fixed for every parity symbol of each of the parity nodes $j \in \mathbb{Z}_r^+$.

\begin{theorem}  There exists a choice of non-zero coefficients from $\mathbb{F}_q$ to symbols from $R_2$ and a common nonzero coefficient $c$ to symbols from $R_1$ such that the code is MDS, if $q \geq {n \choose k}r^{m+1}$.
\end{theorem}


\begin{proof} It follows from the construction that the set $R_2$ corresponding to a parity symbol in any parity node $j \in \mathbb{Z}_r$ is the set of all message symbols lying in the same row. Suppose that the set $R_1$ is an empty set for each parity symbol belonging to any of the parity nodes $j \in \mathbb{Z}_r^+$. In such case, the code will take the form of a vector MDS code obtained by vertically stacking $\alpha$ scalar MDS codes, and such a code indeed exists. Hence it is possible to make an assignment of coefficients to symbols from $R_2$ such that the code becomes MDS. Consider such an assignment, and let $C$ be a codeword for the code. 
	
	Next, consider a set $D$ of $k$ nodes comprising of $k_1$ systematic nodes, and $k_2$ parity nodes. We denote by $C_D$ a $(B \times 1)$-size vector obtained by vectorizing the codeword $C$ restricted to the set of nodes $D$. The vectorization is such that the first $k_1\alpha$ are the messages symbols, and remaining $k_2\alpha$ the parity symbols. Then we have
	\bea
	C_D & = & \left[ \begin{array}{c} A_D \\
		B_D \end{array} \right] M,
	\eea
	where $M$ is the $(B \times 1)$-size message vector, $A_D$ is a matrix of size ${k_1\alpha \times B}$, and $B_D$ is a matrix of size ${k_2\alpha \times B}$. Note that the matrix 
	\bea
	E_D & = & \left[ \begin{array}{c} A_D \\
		B_D \end{array} \right]
	\eea must always be invertible for every choice $D$ of $k$ nodes. Hence its determinant is non-zero.
	
	However since the set $R_1 \neq \phi$ for any parity symbol, each of the symbols in $R_1$ has a coefficient $c$, the matrix $B_D$ will not correspond to the actual MDS code of our interest. So we need to replace $B_D$ with a modified version $\hat{B}_D$ that will satisfy our requirements. For every row of $B_D$, we construct the corresponding row of $\hat{B}_D$ as follows: We choose to keep the non-zero entries as such. Then for every message symbol in $R_1 \setminus R_2$, we populate the corresponding entry in the row with $c$ replacing $0$. Since $R_1 \setminus R_2$ has at most $m$ elements, the number of positions thus modified will be at most $m$. Then we claim that the matrix
	\bea
	\hat{E}_D & = & \left[ \begin{array}{c} A_D \\
		\hat{B}_D \end{array} \right]
	\eea 
	thus obtained is invertible. Consider the determinant of $\hat{E}_D$ as a polynomial $g_D(c)$, in the indeterminate $c$. Clearly, the polynomial must evaluate to a non-zero value for the assignment $c=0$ since otherwise, the determinant of $E_D$ would be zero. Hence $g(c)$ cannot be the zero polynomial. We also have
	\beq
	\text{deg}(g_D(c)) \leq k_2\alpha \leq r\alpha = r^{m+1}
	\eeq
	Next consider the polynomial
	\bea
	h(c) & = & \prod_{D \subset [n], |D|=k} g_D(c).
	\eea
	Cleary $h(c)$ is not identically zero, its degree is upper bounded by ${n \choose k}r^{m+1}$. Hence it is sufficient that we find an assignment for $c$ the evaluates the polynomial $h(c)$ to a non-zero value. This is possible if we choose $q \geq {n \choose k}r^{m+1}$.
	
\end{proof}
	  
\section{Conclusions}	  
\label{discuss}

We presented an alternative construction of an access-optimal code that repairs systematic nodes. The parameter in our construction is $m$. For designing codes with $r$ parity nodes, we will set $k$ to $mr$ and $\alpha$ to $r^m$.  Our construction was inspired by construction of zigzag codes.  A novel feature of our construction is that in our construction, the repair of symbols is carried out sequentially in contrast to parallel repair in the case of zigzag codes. Here, one set of $\alpha/r$ symbols are independently repaired first, but for the rest of the symbols, the previously repaired $\alpha/r$ symbols are also used along with the accessed data from other nodes. 

Since our code has an optimal level of sub-packetization, it will be interesting to investigate whether the level of sub-packetization suffices for the repair of parity nodes as well.  To date, the best known access-optimal construction that can repair both systematic and parity node failure has $\alpha = r^{k+1}$ which is much larger than the achievable bounds in the case of repair of just the systematic nodes.

\bibliographystyle{IEEEtran}
\bibliography{paper_bib}

\end{document}